\documentclass[cleveref, autoref, thm-restate, english]{lipics-v2021}

\usepackage{hchang-socg}
\nolinenumbers

\newcommand{\HH}{H}
\newcommand{\SP}{\pi}

\newcommand{\dist}{d}
\newcommand{\cell}{\mathit{Cell}}
\newcommand{\hit}{\mathit{hit}}
\newcommand{\opt}{\mathrm{opt}}
\newcommand{\pab}[1][v]{\pi_{\alpha,\beta}(#1)}
\newcommand{\base}{\mathit{base}}
\newcommand{\resident}{\mathit{resident}}

\newcounter{ctr}
\loop
 \stepcounter{ctr}
 \expandafter\edef\csname c\Alph{ctr}\endcsname{\noexpand\mathcal{\Alph{ctr}}}
\ifnum\thectr<26
\repeat

\title{Single-Criteria Metric $r$-Dominating Set Problem via Minor-Preserving Support} 

\author{Reilly Browne}{Dartmouth College, NH, USA  \and \url{http://www.myhomepage.edu} }{reilly.browne.gr@dartmouth.edu}{https://orcid.org/0000-0003-3725-5245}{Supported by the U.S. National
Science Foundation CAREER Award under the Grant No. CCF-2443017}

\author{Hsien-Chih Chang}{Dartmouth College, NH, USA  \and \url{http://www.myhomepage.edu} }{hsien-chih.chang@dartmouth.edu}{https://orcid.org/0000-0001-6714-7988}{Supported by the U.S. National
Science Foundation CAREER Award under the Grant No. CCF-2443017}

\authorrunning{Browne and Chang} 

\Copyright{Reilly Browne and Hsien-Chih Chang} 


\category{} 

\relatedversiondetails[linktext={https://doi.org/10.5281/zenodo.19271844}]{Full Version}{https://doi.org/10.5281/zenodo.19271844} 

\hideLIPIcs

\begin{document}

\maketitle

\begin{abstract}
Given an unweighted graph $G$, the \emph{minimum $r$-dominating set problem} asks for a subset of vertices $S$ of the smallest cardinality, such that every vertex in $G$ is within radius $r$ to some vertex in~$S$.
While the $r$-dominating set problem on planar graph admits PTAS from Baker's shifting/layering technique when $r$ is a constant, the problem becomes significantly harder when $r$ can depend on $n$.
In fact, under Exponential-Time Hypothesis, Fox-Epstein \etal~[SODA 2019] observed that no efficient PTAS can exist for the unbounded $r$-dominating set problem on planar graphs.
One may consider even harder 
weighted-variant known as the \emph{vertex-weighted metric $r$-dominating set}, where edges are associated with lengths, and every vertex is associated with a positive-valued weight, and the goal is to compute an $r$-dominating set with minimum total weight.
As a result, people resorted to \emph{bicriteria} algorithms by allowing the returned solution to use radius-$(1+\e)r$ balls instead, in addition to the total weight being a $1+\e$ approximation to the optimal value.

We establish the first \emph{single-criteria} polynomial-time $O(1)$-approximation algorithm for the vertex-weighted metric $r$-dominating set problem on planar graphs when $r$ is part of the input, and can be arbitrarily large compared to $n$.
Our new (single-criteria) $O(1)$-approximation algorithm uses the quasi-uniformity sampling technique of Chan \etal~[SODA 2012] by bounding the \emph{shallow cell complexity} of the (unbounded) radius-$r$ ball system to be linear in $n$.
To this end we have two technical innovations:
\begin{enumerate}
\item The discrete ball system on planar graphs are neither pseudodisks nor have well-defined boundaries for standard union-complexity arguments.
We construct a \emph{support graph} for arbitrary distance ball systems as contractions of Voronoi cells; the sparseness comes as a byproduct.
\item
We present an assignment of each depth-($\geq\! 3$) cell to a unique 3-tuple of ball centers. This allows us to use standard Clarkson-Shor techniques to reduce the counting to cells of depth \emph{exactly} 3, which we prove to be size $O(n)$ by a novel geometric argument based on our support being a Voronoi contraction.
\end{enumerate}
\end{abstract}

\newpage
\section{Introduction}

The concept of \emph{distance balls}
is of fundamental importance to any given metric space $(X,\dist)$ as the concept of \emph{open sets} to any topological space.
A \EMPH{distance ball $B(v,r)$} in a metric space $(X,\dist)$ with radius $r$ centered at point $v$ is the subset of points $\set{u \in X : \dist(u,v) \le r}$ in $X$.
Many of the algorithmic success on graphs can be attributed to new understandings of the structure of distance balls in the corresponding graph metric.

\paragraph{The role of distance balls.}
Take the \emph{minimum dominating set}---one of the 21 famous NP-complete problems of Karp---as an example.  
Given an unweighted graph $G$, the \EMPH{minimum dominating set problem} (\textsc{Dominating Set} for short) asks for a subset of vertices $S$ of smallest cardinality, such that every vertex in $G$ is either in $S$, or is adjacent to some vertex in $S$.
\textsc{Dominating Set} can also be interpreted as an instance of the \emph{set cover problem}:
Consider the set $\cR \coloneqq \set{B(v,1) : v \in V_G}$ of 1-neighborhood balls.
Observe that a set $S$ is dominating in $G$ if and only if every vertex is covered by the subset 
of 1-neighborhood balls in $\cR$ centered in $S$.
In other words, a minimum solution of the set-cover instance $(V_G,\cR)$ is a minimum dominating set in $G$.
This simple observation immediately implies an $O(\log n)$-approximation to \textsc{Dominating Set}, where $n$ is the number of nodes in $G$, by the standard greedy set cover algorithm.
$O(\log n)$-approximation is the best possible (under popular beliefs in complexity) if we make no further assumptions on the graph~\cite{fei-tlasc-1998,cc-ahdsp-2008}; this is essentially because the interaction between 1-neighborhood balls in general graphs can be complex and unruly.

What if we consider special classes of graphs with additional geometry constraints?  
One of the most well-studied examples is \EMPH{planar graphs}, graphs with a drawing in the plane where no two edges cross each other.
Planar graphs exhibit many helpful traits for algorithm design; particularly, several partitioning tools exist, including \emph{separators}~\cite{lt-stpg-1979,mil-fsscs-1986,tho-corad-2004}, \emph{low-diameter decompositions}~\cite{kpr-emndm-1993,bar-pamsi-1996,blt-scpgg-2014}, and an algorithmic paradigm now famously known as \emph{Baker's shifting/layering technique}~\cite{bak-aanpp-1994} which utilizes diameter-treewidth property~\cite{epp-dtmgf-2000} and the decomposability of planar graphs into $k$-outerplanar graphs.  
The idea is to compute a BFS-tree from an arbitrary node, then focus only on a few consecutive layers at a time.
If the optimization problem at hand is \emph{local} (for \textsc{Dominating Set}, node coverage can be checked by examining immediate neighbors), 
then a major portion of the optimal solution would be preserved within the consecutive layers.  
By iterating over the initial shift value and removing one out of every $O(1/\e)$ layers in the BFS-tree, one may guarantee a $(1+\e)$-approximation to the optimal solution.
This is the main idea behind the polynomial-time approximation scheme (PTAS) of \textsc{Dominating Set} on unweighted planar graphs~\cite{bak-aanpp-1994}.

Not surprisingly, the same layering technique can be applied to the \EMPH{$r$-dominating set problem} ($r$-\textsc{Dominating Set} for short) in planar graphs for any absolute constant $r$, where the goal is similar to \textsc{Dominating Set} but every node in $S$ can ``cover'' a vertex subset of radius $r$~\cite{dfht-fakpg-2005}.
$r$-\textsc{Dominating Set} was introduced by Barilan, Kortsarz, and Peleg~\cite{bkp-hanc-1993} and later rediscovered in several other context~\cite{gprs-skspg-2001,dfht-fakpg-2005}.
%
However, one encounters some trouble when generalizing Baker's technique to minor-free graphs.  
This is because arbitrary minor-free graph no longer satisfies the diameter-treewidth property (imagine adding one extra vertex connecting to every node in a planar graph; the diameter becomes 2 but the graph is $K_6$-minor-free).
One possible way to overcome the issue is to apply the Robertson-Seymour structural theorem for minor-free graphs~\cite{gro-ltema-2003,rs-gmxen-2003}, which seems to be overkill.
Another, possibly more satisfying solution, is to understand what the $O(1)$-radius neighborhood balls can look like in a minor-free graph.
This is the route taken by the \emph{local search} approach.
On planar graphs, typical local search analysis~\cite{cg-spfge-2015} is usually done by applying an \emph{exchange argument} on the base graph $G$, utilizing the fact that planar graphs have small separators~\cite{lt-stpg-1979} (in the form of an $r$-division~\cite{fre-faspp-1987}).
Har-Peled and Quanrud~\cite{hq-aaplg-2017a} observed that, for problems like $r$-dominating set, the exchange argument in a typical local search analysis can be performed on the \emph{intersection graph} of $O(1)$-neighborhood balls~\cite{mr-irghs-2010,ch-aamis-2012} (more accurately, the subgraph $G_\cR$ induced by the union of the local search solution and the optimal solution).
They further observed that the intersection graph $G_\cR$ over $O(1)$-radius balls in minor-free graph $G$, while not necessarily itself minor-free, must have \emph{polynomial expansion}~\cite{no-gcbed-2008}.
Without going into definitions, the class of graphs with polynomial expansion coincide with those that have sublinear-size separators for all subgraphs~\cite{dn-ssspe-2015}.
Thus one can apply the same kind of exchange argument using separator decomposition on the ball intersection graph $G_\cR$.

\paragraph{Arbitrary radii and arbitrary weights.}
The $r$-dominating set problem becomes significantly harder when $r$ can depend on $n$~\cite{ekm-akpg-2014}.  
In the literature this is usually refer to as the the \EMPH{metric/unbounded $r$-dominating set problem}~\cite{fl-cetlt-2021,fl-ltepm-2022}, where we further allow arbitrary nonnegative \EMPH{edge-length $\ell$} on the input graph, in addition to the covering radius $r$ (now measuring with respect to $\ell$) becoming part of the input.
We can even consider a weighted variant of the problem, where every vertex is associated with a positive-valued \EMPH{weight $w$}, and the goal is to compute an $r$-dominating set with minimum total weight (instead of minimum cardinality).
Baker's layering technique no longer applies, nor does the local search strategy for the most part.
(There is an exception~\cite{le-sraam-2018,fl-cetlt-2021} which we will talk about in the technical section.)
This is somewhat expected because the unbounded $r$-dominating set problem is no longer ``local'', and the intersection graph of radius-$r$ balls may be dense and have complex interactions.
In fact, under Exponential-Time Hypothesis,
Fox-Epstein \etal~\cite{mp-opapf-2015,fks-epglg-2019} proved that no efficient PTAS can exist for the unbounded $r$-dominating set problem on planar graphs.
Thus, recent work has resorted to \emph{bicriteria} algorithms, allowing the returned solution to use a dominating set of radius $(1+\e)r$ instead, in addition to the total weight being a $1+\e$ approximation to the optimal value~\cite{ekm-akpg-2014,fks-epglg-2019}.
The idea is to embed the original planar/minor-free graph metric into some other \emph{tree-like metrics}, where the $r$-dominating set problem can easily be solved using dynamic programming.
Unfortunately one has to suffer some form of distortion (which translates to the slack in covering radius); the best result so far gives
$+\e\Delta$-additive embedding of diameter-$\Delta$ minor-free graphs into graphs of $O_{\e}(1)$ treewidth~\cite{ccl+-cpmot-2023,ccl+-spmgs-2024}.
This line of research has been fruitful and ultimately leads to bicriteria EPTAS for planar graphs~\cite{fks-epglg-2019} and bicriteria Q(Q)PTAS for minor-free graphs~\cite{fl-cetlt-2021,fl-ltepm-2022}.

But what if we really want to solve the unbounded $r$-dominating set problem without compromise on the radius $r$?
The lower bound result~\cite{mp-opapf-2015,fks-epglg-2019} leaves us no choice but to look at $O(1)$-approximation (or at least, $(1+\e)$-approximation with running time $n^{f(\e)}$). 
The existence of such algorithm was explicitly asked by Filtser and Le~\cite{fl-ltepm-2022}.

\paragraph{Our contribution.}
We establish the first polynomial-time $O(1)$-approximation algorithm for the (vertex)-weighted metric $r$-dominating set problem for unbounded $r$.
We emphasize that $r$ is part of the input, and can be arbitrarily large compared to $n$.

\begin{theorem}
\label{thm:weighted-dom}
Let $G$ be an arbitrary $n$-node planar graph, where each vertex $v$ is associated with a weight $w(v)$.
Let $r$ be an arbitrary radius parameter.  
There is a polynomial-time algorithm that computes an $r$-dominating set in $G$ whose weight is $O(1)$ times the~optimum.
\end{theorem}

\subsection{Techniques}

Our new (single-criteria) $O(1)$-approximation algorithm for weighted $r$-dominating set on planar graphs uses the quasi-uniformity sampling technique, originally designed for weighted set cover problem on geometric shapes~\cite{var-wgscq-2010,cgks-wcpgs-2012}.
But first, we have to describe a bit of (technical) history and two immediate challenges when adapting to our setting.

\paragraph{Challenge 1: Planar support.}
Geometric set cover problems have their own history and success~\cite{bg-aoscf-1995,ClarksonShor1989,ClarksonVaradarajan2007,var-enuc-2009}.
Notably, for an unweighted set cover instance $(X,\cR)$, 
the result of Br\"onnimann and Goodrich~\cite{bg-aoscf-1995} provided a recipe to 
achieve $O(1)$-approximation.
Their pioneering use of multiplicative weight update algorithm requires the existence of small \emph{$\e$-net}:
First, find an $\e$-net $N$ for the dual set system $(\cR,X)$ (notice that here the $\e$-net is in the dual and thus a subset of ranges in $\cR$). 
If there are element $x$ in $X$ not covered by $N$,
\emph{double} the weight of every range in $\cR$ that covers $x$, to increase the chance of choosing one such range in the next round.
As a result, by setting $\e = \Theta(1/\opt)$, this implies an $O(g(\opt))$-approximation for any set cover instance, given the dual system has an $\e$-net of size $O(\frac{1}{\e} \cdot g(\frac{1}{\e}))$.  
(See \cite{ers-hswvs-2005} for an LP interpretation and a simpler proof.)

For set system of bounded VC-dimension~\cite{vc-ucrfe-1971}, there is always an $\e$-net of size $O(\frac{1}{\e} \log \frac{1}{\e})$~\cite{hw-esrq-1987}.  Also, if the VC-dimension of a primal system is $d$, then the dual system has VC-dimension at most $2^d$~\cite{har-gaa-2011}.
%
%
%
To get $O(1)$-approximation, however, we need a much stronger assumption of $O(\frac{1}{\e})$-size $\e$-net \emph{in the dual system}.
This is not always achievable even for simple geometric shapes like lines and rectangles~\cite{mv-ee-2017}; in fact, among the basic geometric shapes, only a handful are known to have linear-size $\e$-net, including halfspaces in 2D and 3D~\cite{msw-hnlls-1990,pw-nbe-1990,mat-rph-1992}, unit-cubes in 3D~\cite{ClarksonVaradarajan2007}, and family of pseudodisks in the plane (both primal and dual)~\cite{PR08}.

For our application of unbounded $r$-dominating set problem, the ``geometric shapes'' are the radius-$r$ balls in planar graphs, which is discrete in nature, lacking well-defined boundary curves to even make the definition of pseudodisks meaningful. 
One might sneer at the issue and say, ``there must be a way; how can distance balls in the plane be much different from pseudodisks?''
But even if we redefine the notion of pseudodisks combinatorially (for every pair of regions $A$ and $B$, all three sets $A \cap B$, $A \setminus B$, and $B \setminus A$ are connected), the distance balls are still not pseudodisk. In fact, they can be \emph{piercing} in the sense of Raman-Ray~\cite{RamanRay2020}, as illustrated by Figure~\ref{fig:piercing}.

\begin{figure}[h!]
    \centering
    \includegraphics[width=0.5\linewidth,page=8]{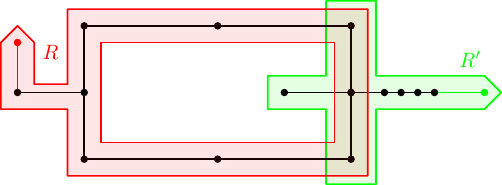}
    \caption{\internallinenumbers
    Two distance balls $R$ and $R'$ that are piercing. Here both $R$ and $R'$ have radius 6.}
    \label{fig:piercing}
\end{figure}

Fortunately, Pyrga and Ray~\cite{PR08} provided an alternative method to obtain linear-size $\e$-net (in the dual) --- to construct a \emph{sparse support} for the primal system $(X,\cR)$.
We say a set system $(X,\cR)$ has a \EMPH{support}, if for any subsystem $(X',\cR_{|X'})$ of $(X,\cR)$, there is an undirected graph $\HH$ on the elements of $X'$ such that for every set $R$ in $\cR_{|X'}$, the subgraph $\HH_R$ of $\HH$ induced by the elements of $R$ is connected.
The support is \EMPH{sparse} if every such graph $\HH$ is sparse.
The work of Pyrga and Ray~\cite{PR08} implies that if $(X,\cR)$ has bounded VC-dimension and a sparse support, then the dual system $(X,\cR)$ has $\e$-net of size $O(\frac{1}{\e})$.

Our first technical contribution is to compute a support for any collection of distance balls in any \emph{arbitrary} graph (not just planar graphs). Crucially, if the given graph is planar, then so is the support.
We remind the reader that our original goal is to solve the $r$-dominating set problem with unbounded $r$; thus the construction of the support has to hold for balls of unbounded radii as well.
This is the first construction of its kind; most previous work on sparse support focus on geometric shapes, and almost entirely restricted to \emph{planar support} for set systems that are pseudodisks~\cite{PR08} and their generalizations~\cite{RamanRay2020}.
%
For piercing systems like axis-parallel rectangles, there are configurations
where no planar support can exist~\cite{prrs-facps-2024}. 
%

Instead of focusing on planarity, 
we shift the narrative by arguing that a support graph for $(X,\cR)$ is really a \emph{connectivity sketch} that exists for every graph $G$ on $X$, as long as the sets in $\cR$ are \emph{distance balls} on $G$.
%
Our construction for support on arbitrary balls in $G$ can be described simply as follows:
Build a \emph{Voronoi diagram} of $V(G)$ with respect to balls centers, then contract every Voronoi cell into a single node.  We argue that the contracted graph is a support of $G$.
(The benefit of this additional property becomes evident when we tackle the next challenge.)
Both planarity and sparseness come as a byproduct of the fact that 
contraction of the Voronoi diagram is a \emph{minor} of $G$.  (See Theorem~\ref{thm:dual_support} for a precise statement.)  
(A similar Voronoi diagram idea can be found in the work of Le~\cite{le-sraam-2018,fl-cetlt-2021}; but their version only enables single-criteria algorithm for \emph{unweighted} metric $r$-dominating set problem.)

\paragraph{Challenge 2: Linear shallow cell complexity.}
There is another, perhaps more severe, obstacle to overcome.
We want to solve the \emph{weighted} version of the $r$-dominating set problem.
The Br\"onnimann-Goodrich MWU approach relies on computing regular $\e$-net (as subset of balls covering the vertices); the \emph{size} bound on $\e$-net translates to approximation guarantee on \emph{size} of the cover, and thus can only work in the unweighted setting.
Varadarajan~\cite{var-wgscq-2010} extended the $\e$-net idea to the weighted setting.
This \emph{quasi-uniformity sampling technique} computes $\e$-nets that samples each set in $\cR$ with (roughly) equal probability, which then can be made aware of the weights on $\cR$.
Chan, Grant, K\"onemann, and Sharpe~\cite{cgks-wcpgs-2012} later polished this technique and made connection to a more refined measure on set systems called \EMPH{shallow cell complexity}. 
In our context, cell complexity refers to the number of different equivalent classes (called \emph{cells}) of points that are covered by the same subset of $r$-balls.  We say a cell has \emph{depth $k$} if the points within are covered by up to $k$ balls.  
A set system has shallow cell complexity $f(|\cR|,k)$ if the number of cells of depth at most $k$ is bounded by $f(|\cR|,k)$;
for the purpose of obtaining $O(1)$-approximation for set cover, one requires the shallow cell complexity to be $O(|\cR| \cdot \poly k)$.

The quasi-uniformity sampling technique can be used to solve the \emph{constant $r$} case for weighted $r$-dominating set problem on minor-free graphs (and graphs with bounded expansion as well~\cite{dvo-amcbe-2022}).
This is because in the case of constant radius $r$,
the \emph{whole} cell complexity is linearly bounded~\cite{rvs-cbenc-2019,jr-ncpg-2023}, which directly bounds the shallow counterpart.
Unfortunately, it is known that the cell complexity of $r$-balls on a planar graph can be cubic in $r$~\cite{jr-ncpg-2023}, and thus cannot help in the unbounded $r$ setting.

One way to see why halfspaces and pseudodisks are special is that they both have linear \emph{union complexity}~\cite{klps-ujrct-1986,ClarksonVaradarajan2007}.
Application of the classical Clarkson-Shor technique then implies that the shallow cell complexity is bounded by $O(|\cR| \cdot k)$.
%
When working with arbitrary-radius balls on graphs however, we run into trouble applying Clarkson-Shor.
First, there is no well-defined notion of union complexity for a collection of subgraphs.
One would imagine the support can take the role as the linear-size object at the base case for Clarkson-Shor; but this intuition is insufficient.
%
The main issue is that for pseudodisks, the number of cells can be charged to the complexity of the whole pseudodisk arrangement by Euler's formula, and every crossing in the arrangement is uniquely defined by two pseudodisks.
There is no immediate equivalent object for distance balls that allows for such counting.


Our solution is to reimagine the unique representative argument by Raman and Ray for the non-piercing family~\cite{rr-gsmp-2022}; the new proof is substantially different and relies on geometric properties of distance balls, instead of topological properties of curve sweeping~\cite[\S6]{rr-gsmp-2022}.
First we argue that every cell can be charged to a unique 3-tuple of ball centers (\Cref{lma:encoding}).  
This allows us to apply Clarkson-Shor and reduce the problem to counting cells of depth 1, 2, and 3.
Depth 1 and 2 are straightforward; for depth-3 cells, we exploit the additional property that our support for distance balls came from contracting \emph{Voronoi cells} with respect to the ball centers, and thus the number of depth-3 cells can be counted towards the complexity of the Voronoi diagram, which is asymptotically the same as the planar support (\Cref{lma:depth_3_linear}).

\section{Terminology and Support Graphs}

Let's first introduce the terminology needed for concrete discussion.
Our support construction actually works for balls in an arbitrary \emph{directed} graph. 
We refer to the digraph that our ball system lives on as the \EMPH{underlying graph} $G$. 
Note that all results apply to balls in undirected graphs.
%
A \EMPH{ball} $D$ can be defined using a pair $(c_D,r_D)$ where $c_D$ is a vertex of $G$ called the \EMPH{center}, and a \EMPH{radius} $r_D$ which is a real number. 
Ball $D$ is then defined to be the set of all vertices $v$ whose shortest path from $c_D$ has length at most $r_D$. 

Given a system $(U, \cR)$, the \EMPH{dual system} $(\cR,U)$ has the elements of $U$ treated as subsets of $\cR$: if element $e$ in $U$ is a member of set $R$ in $\cR$ in the primal system, then $R$ is a member of \emph{set} $e$ in the dual system.  
We use the term \EMPH{dual support} of $(U,\cR)$ to refer to a support graph of the dual system $(\cR,U)$. 
(We give an example of a dual support in Figure~\ref{fig:support_def}.)
Although our $r$-dominating set application demands primal support,
the set system of balls \emph{of equal radius} is self-dual,
Thus, we do not differentiate the support for the primal from the dual, but instead describe the support for the dual system since it is conceptually easier. 
%

\begin{figure}[t!]
    \centering
    \includegraphics[page=1, width=0.6\linewidth]{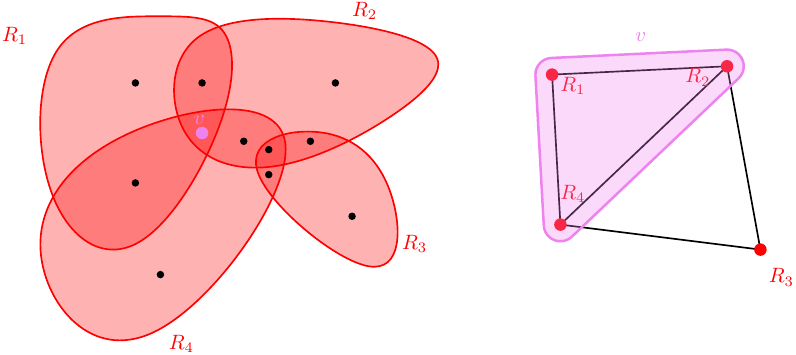}
    \caption{\internallinenumbers
    A set system with a dual support. The red regions containing the violet vertex $v$ form a connected subgraph in the support. 
    }
    \label{fig:support_def}
\end{figure}


In \Cref{sec:dual} we construct a dual support for any system of balls, by contracting cells of the Voronoi diagram of the underlying graph (with respect to ball centers) into a~minor.

\begin{restatable}[Support]{theorem}{dualsupport}
\label{thm:dual_support}
    Any system of balls $\cR$ on an underlying digraph $G$ has a dual support $\HH(\cR)$ such that if $G$ does not contain a minor $K$ with minimum vertex degree two, then neither does $\HH(\cR)$. 
\end{restatable}

\noindent \textit{Remark.}  
Unfortunately the min-degree requirement cannot be dropped in general.  For example, if $G$ is $P_2$-minor-free, we have a set of disjoint vertices.  But two balls can centered on the same vertex, and thus connected in the dual support, which implies the existence of an edge and so the support is not $P_2$-minor-free.
Still the restriction is not severe as most of the interesting minor-free graph classes satisfy the min-degree condition, including $K_h$-minors for $h\ge 3$, grid minors, and all forbidden minors for bounded-genus graphs.

\medskip
Since planar graphs can be classified by their excluded minors, these results imply a planar dual support for the ball system $\cR$ on $G$. 
We then proceed to show that $\cR$ has linear shallow cell complexity in \Cref{S:shallow-cell}. 

\begin{restatable}[Shallow Cell Complexity]{theorem}{shallowcell}\label{thm:shallow_cell_planar}
    Let $G$ be a planar graph.
    Any system of balls $(V_G,\cR)$ has shallow cell complexity $O(|\cR| \cdot k^2)$.
\end{restatable}




\subsection{Applications}
In addition to the result on computing unbounded $r$-dominating set in planar graphs, the supports also have their own implications. 
An important property of our support construction is that we only ever contract edges that are within radius $r$ of a ball center. 
This means that if $G$ is an unweighted graph with \emph{polynomial expansion} and we are considering constant-radius balls, the support graph will also have polynomial expansion, allowing us to use the sublinear separators from Har-Peled and Quanrud~\cite{hq-aaplg-2017a} to get the same PTAS results for unweighted problems as Raman and Ray~\cite{RamanRay2020} achieved for non-piercing regions. 
If $r$ is not bounded by a constant, we can achieve the same when $G$ is minor-free~\cite{minorsep}.

Another natural use of supports is for problems relating to geodesic disks in a polygonal domain (possibly with holes). 
It is known that planar metrics and polygonal domains 
are \emph{equivalent}: any system of geodesic disks in a polygonal domain in $\mathbb{R}^2$ can be represented as a system of balls in a planar graph, and vice versa~\cite{SujoyMetric}. 
By the equivalence,
our techniques naturally give planar supports for geodesic disks in polygonal domains.






\subsection{Related Work on Support Graphs}


%
Van Cleemput~\cite{van1976hypergraph} and Voloshina and Feinberg~\cite{VoloshinaFeinberg1984} used the existence of a {planar} support graph to give meaningful notion of planarity in the hypergraph setting. 
As such, support graphs have seen interest in (hyper)graph drawing community~\cite{JohnsonPollak,KaufmanSubdivision,buchin2009planar}.
In the realm of optimization, Pyrga and Ray~\cite{PR08}
%
%
used support graphs to construct linear-sized $\e$-nets for several set systems with bounded VC dimension, including halfspaces in $\mathbb{R}^3$ and pseudodisks (more generally, $r$-admissable regions) in $\mathbb{R}^2$.  Here the support graphs were required to be \emph{sparse} (i.e.\ with linearly many edges) in order to bootstrap existing $\e$-net sizes down to $O(\frac{1}{\e})$. 
%
Later on support graphs that are \emph{planar} were used by Mustafa and Ray~\cite{mr-irghs-2010} to show a PTAS for various optimization problem in the geometric setting.
Instead of $\e$-net, they exploited the existence of \emph{sublinear-size} separators in planar graphs to show that an $O_\e(1)$-swap local search algorithm gives a $(1+\e)$-approximation for minimum geometric hitting set. 
The same paradigm of analyzing local search with separators on planar support was extended to other problems, such as maximum independent set and generalized set cover (including dominating set and hitting set) in pseudodisks and non-piercing regions~\cite{ch-aamis-2012,durocher2015duality,govindarajan_packing_nonpiercing,RamanRay2020}, terrain guarding~\cite{gkkv-gtls-2014,bb-elsag-2017}, and covering the boundary of a simple orthogonal polygon with rectangles~\cite{basuroy}. 
(A survey on the many use of support can be found in Raman and Singh~\cite{raman2023hypergraph}.)

\section{Existence of Support Graph}
\label{sec:dual}

To begin constructing the dual support for a system of balls $\cR$ in a graph $G$, we will first modify $G$ slightly and reduce to the case where \emph{every ball in $\cR$ has the same radius}.
We will refer to the new graph as \EMPH{$G'$}.
(Recall we work with the general case of digraphs.)
We augment digraph $G$ by adding a new node corresponding to each ball in $\cR$. 
(For clarity, we refer to original vertices of $G$ as \emph{vertices} and new vertices corresponding to $\cR$ as \emph{nodes}.)
For each ball $R$ in $\cR$, we create a single node, \EMPH{$x_R$}, which has only one edge: an outgoing edge into $c_R$, the center of the ball $R$. 
We give the edge a weight of $r_{\max} - r_R$, where \EMPH{$r_{\max}$} is defined to be $\max_{R'\in \cR} r_{R'}$. 
This allows us to instead consider each $R$ in $\cR$ to be centered at $x_R$, all with the same radius $r_{\max}$, and the balls will still define the same set system $\cR$ over the vertices in $G$. 
Finally, remove all vertices which are not contained in any ball of $\cR$ and slightly perturb the edge weights (without changing any ball containment relations) so that 
every vertex $v$ has a unique closest ball center in $\set{x_R : R \in \cR}$.

To build the dual support, we construct a \emph{Voronoi partition} of $G'$ with respect to the set of nodes. 
A \EMPH{Voronoi partition $\cF$} of a graph $G'$ with respect to a set of nodes $S$ is a partition of $V_{G'}$ into subsets such that for every node $R \in S$, there is a corresponding subset $\text{\EMPH{$F_R$}} \in \cF$ which contains all vertices $u$ in $G'$ for which $d(x_R,u) \leq \min_{x_{R
}\in S \setminus \set{x_R} } d(x_{R'},u)$. 
For clarity, we will refer to each set $F_R$ in $\cF$ as the \EMPH{Voronoi cell} of $R$. 
In our case, we construct the Voronoi partition $\cF_\cR$ of $G'$ with respect to $\set{x_R : R \in \cR}$, which we identify as $\cR$. 
Each cell~$F_R$ in the Voronoi partition is a connected subset of $G'$ because for every vertex $v$ in $F_R$ reachable from $x_R$, the vertices on the shortest path between $x_R$ and $v$ must belong to $F_R$ as well.
%
%
%
Thus we can simply contract every $F$ in $\cF_\cR$ to a (contraction)-minor on only the nodes in $\cR$. This final graph \EMPH{$\HH$}, we claim, is a dual support for $(G,\cR)$.


\begin{figure}[ht]
\centering\includegraphics[width=0.7\linewidth,page=3]{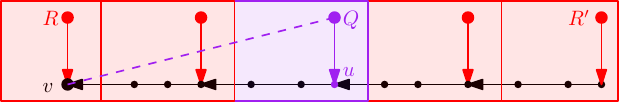}
    \caption{\internallinenumbers
    An illustration of the proof to Lemma~\ref{lma:dual_is_a_support} that $\pi$ can only pass through Voronoi cells of $\cR_v$.  The purple region is the Voronoi cell $F_Q$}
    \label{fig:dual_lemma}
\end{figure}

\begin{lemma}\label{lma:dual_is_a_support}
    For any vertex $v$ in $G'$, the set of balls $\EMPH{$\cR_v$} \subseteq \cR$ that contain $v$ must induce a connected subgraph in~$\HH$.
\end{lemma}

\begin{proof}
    Consider any vertex $v$ which is  contained in at least two balls in $\cR$. Vertex $v$ must be in the Voronoi cell $F_R$ of some ball $R$ by construction of $G'$. 
    We claim that for any other ball $R'$ of $\cR$ that also contains $v$, there is a path from $x_{R'}$ to $v$ in $G'$, denoted as \EMPH{$\SP$}, that only passes through Voronoi cells corresponding to balls that contain $v$. 
    
    Since we have removed all vertices which have no incoming path from any ball, every vertex along $\pi$ in $G'$ must be in some cell. 
    %
    Assume that some vertex $u$ along $\pi$ is contained in some Voronoi cell $F_{Q}$ corresponding to a ball $Q$ that does not contain $v$. 
    By definition of Voronoi cell, we have that $d(x_{Q}, u) \leq d(x_{R'},u)$. By the definition of shortest paths, we know that $d(x_{R'},v) = d(x_{R'},u) + d(u,v)$ . By triangle inequality, we then have the following:
    \[
    d(x_{Q},v) \leq d(x_{Q},u) + d(u,v) \leq d(x_{R'},u) + d(u,v) = d(x_{R'},v).
    \]
    But, $d(x_{R'}, v) \leq r_{\max}$ and $d(x_{Q},v) > r_{\max}$ since $v$ is in $R'$ but not in $Q$,
    contradiction.
    %
\end{proof}


\begin{proof}[Proof of Theorem~\ref{thm:dual_support}]
    Construct $\HH$ as described above. 
    By Lemma \ref{lma:dual_is_a_support}, we know that $\HH$ satisfies the definition of a dual support for system $(V_G,\cR)$. 
    It remains to show that if $G$ has no minor $K$ with minimum degree two, then neither can $\HH$. 
    This is done by arguing that $G'$ has no such minors, and then the fact that $\HH$ is a minor of $G'$ concludes the proof.

    When constructing $G'$, we only perform two types of operations that would affect the existence of minors: deleting vertices and adding in nodes $\set{x_R}$ of degree one. 
    Let $K$ be a minor model (as a collection of disjoint vertex subsets called \EMPH{supernodes}) of $G'$ that has minimum degree two. 
    If $K$ only contains vertices but not nodes in $G'$, then $K$ must appear in $G$ as well (even though $G'$ was obtained by removing some vertices from $G$).
    So 
    at least one supernode $\eta$ of $K$ must contain some degree-1 $x_R$ that was added in $G'$.

    For an edge $(\eta,\eta')$ to exist in $K$, there must be a vertex/node $u \in \eta$ and a vertex/node $u' \in \eta'$ such that $(u,u')$ is an edge of $G'$. For $\eta$ containing $x_R$, $x_R$ must not be the only vertex/node in $\eta$, since $\eta$ have degree at least 2. 
    Because $\eta$ is connected, the neighbor of $x_R$ is also in $\eta$ and the (sole) edge incident on $x_R$ must be contracted. 
    This, however, means that there can be no edge $(x_R,u')$ in $G'$ that corresponds to an edge $(\eta,\eta')$ in $K$, so there is a minor $K'$ isomorphic to $K$ where $x_R$ was simply deleted and thus $x_R$ is not in $\eta$. 
    Since this holds for all $x_R$, there must be some minor of $G$ isomorphic to $K$,
    a contradiction.
    %
\end{proof}



\section{Shallow Cell Complexity of Balls in Planar Graphs}
\label{S:shallow-cell}


Shallow cell complexity was introduced by Chan \etal~\cite{cgks-wcpgs-2012} as a combinatorial analog of union complexity.
We introduce the terminology within our context of distance balls on graphs.

Define the \EMPH{hit set} of a vertex $v$ with respect to $\cR$ to be the subset of balls in $\cR$ which contain $v$.
A \EMPH{cell} of the system $(V_G,\cR)$ is a set of vertices in $V_G$ that have the same hit set of incident balls. 
\emph{Cell complexity} is simply the number of cells in $(V_G,\cR)$, where as \emph{shallow} cell complexity counts the number of cells of a particular \EMPH{depth $k$} in $\cR$, denoted \EMPH{$\cell_k$}, the equivalence classes of vertices whose hit set has size $k$. 
We say the set system has \EMPH{shallow cell complexity $f(n,k)$} if for depth $k$, there are at most $f(n,k)$ cells.
Our goal for the section is to show that for a planar graph $G$, the shallow cell complexity for any system of balls $(V_G,\cR)$ is $O(|\cR|\cdot k^2)$ (\Cref{thm:shallow_cell_planar}). 
%
Just as in the support constructions in \S\ref{sec:dual}, we again reduce to the case where \emph{every ball in our system $\cR$ has the same radius}.


\begin{proof}[Proof of \Cref{thm:shallow_cell_planar}]
We use two key lemmas proven later in the subsections. 
The first claims that for every cell $c$ of depth $k$, there is a unique encoding of $c$, $\langle \alpha(c), \beta(c), \gamma(c) \rangle$, where each member of the tuple is a ball of $\cR$ (Lemma~\ref{lma:encoding} in Section~\ref{S:unique-encoding}). Importantly, the encoding is unique among cells of depth $k$, not universally.  
This naively implies that there is at most $O(|\cR|^3)$ cells of depth $k$. Using the standard Clarkson–Shor arguments argument, we can improve this to $O(|\cR| \cdot k^2)$ if the number of cells of depth (at most) $3$ is $O(|\cR|)$ (Lemma~\ref{lma:depth_3_linear} in Section~\ref{S:depth-3-cells}). 
More precisely: Sample the regions of $\cR$ with probability $1/k$.
A cell $c$ with unique encoding 
appears in the sample as depth-3 cell with probability at least $(ek)^{-3}$.
The expected number of depth-3 cells in the sample is at most $O(|\cR|/k)$, so the number of depth-$k$ cells in $\cR$ is at most $O(|\cR| \cdot k^2)$.
\end{proof}


    

\subsection{Unique 3-Site Encoding of Cells}
\label{S:unique-encoding}

We begin by assigning each cell of $\cell_k$ an \emph{unique} encoding for each fixed depth $k \ge 3$. 
This encoding will depend on the set of incident balls corresponding to $c$. Specifically, the encoding will be formed by a 3-tuple of regions in $\cR$, which we will refer to as \EMPH{$\langle \alpha, \beta, \gamma \rangle$}. 
In assigning the encoding, we do not modify to the set system at all. Instead, we fix an arbitrary vertex \EMPH{$v_c$} for cell $c$ as its \EMPH{representative}. 
We denote the cell containing $v$ as \EMPH{$\cell(v)$}.

For a cell $c$, we begin by designating its \emph{$\seq{\alpha, \beta}$-type}. 
We define \EMPH{$\alpha(v)$} as the furthest ball from $v$ that contains $v$. 
Similarly, we define \EMPH{$\beta(v)$} as the second furthest ball from $v$ that contains $v$. 
For cell $c$, its \EMPH{$\seq{\alpha, \beta}$-type} is simply $\seq{\alpha(v_c), \beta(v_c)}$.
(Notice that $\alpha(v_c)$ and $\beta(v_c)$ might be different if one chooses a different $v_c$ from cell $c$, and thus the $\seq{\alpha, \beta}$-type of $c$ depends on its representative.)
Let \EMPH{$\pab[v]$} be the union of $\SP(\alpha(v), v)$ and $\SP(\beta(v), v)$, 
where \EMPH{$\SP(x,y)$} is the shortest path from $x$ to $y$ in $G$.

\begin{figure}[ht]
    \centering
    \includegraphics[width=0.8\linewidth,page=9]{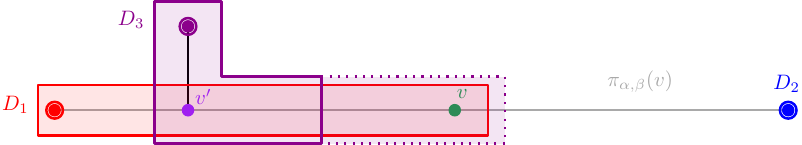}
    \caption{\internallinenumbers
    Given two vertices $v,v'$ satisfying $\alpha(v) = \alpha(v') = D_1$, $\beta(v) = \beta(v') = D_2$ and $v' \in \SP(D_1,v) $, any ball $D_3$ containing $v'$ must also contain $v$.}
    \label{fig:ancestor subsets}
\end{figure}

\begin{lemma}\label{lma:ancestor_subsets}
    For a vertex $v$ with $\alpha(v) = D_1$ and $\beta(v) = D_2$, every vertex $v'$ in $\pi_{\alpha,\beta} (v)$ with $\alpha(v') = D_1$ and $\beta(v') = D_2$ satisfies $\hit(v') \subseteq \hit(v)$. 
\end{lemma}
\begin{proof}
    Assume the contrary.
    Then, there is some vertex $v'$ which is contained within some ball $D_3$ such that $v \not \in D_3$. 
    By definition of $\alpha$ and $\beta$, it must be that the center of $D_3$ is closer to $v'$ than $D_1$ and $D_2$. However, since $v'$ is either in $\SP(c_{D_1}, v)$ or $\SP(c_{D_2}, v)$, this means that $D_3$ is closer to $v$ than one of $D_1$ or $D_2$, and thus $D_3$ must contain $v$,
    a contradiction by the assumption that all balls have equal radius.  
    (See Figure~\ref{fig:ancestor subsets}.)
\end{proof}


The third component of our encoding, \EMPH{$\gamma$}, can now be defined. 
For $k=3$, $\gamma$ is simply the remaining ball after removing $\alpha$ and $\beta$. For $k \geq 4$, we define $\gamma$ procedurally after fixing $\alpha, \beta$.
Denote the set of all cells of depth $k$ with $\alpha(v_c) = D_1$ and $\beta(v_c) = D_2$ as \EMPH{$C_k(D_1,D_2)$}. We additionally collect the representatives $v_c$ of each $c \in C_k(D_1,D_2)$ into a corresponding set \EMPH{$V_k(D_1,D_2)$}.
%
Let \EMPH{$T_{D}(S)$} be the shortest path tree from the center $x_D$ of a ball $D$ to a set of vertices $S$.
Consider the shortest path tree $T_{D_1} (V_k(D_1,D_2))$. 
We will process all $v_c \in V_k(D_1,D_2)$ using the plane drawing of $T_{D_1}(V_k(D_1,D_2))$. 

Perform a depth-first traversal starting at $x_D$, processing children in clockwise order. 
The depth-first traversal $T_{D_1}(V_k(D_1,D_2))$ of  gives a total ordering of the vertices in $V_k(D_1,D_2)$, which also induces a total ordering $\sigma$ on the corresponding cells $C_k(D_1,D_2)$. 
For brevity, we will index the cells based on their order in $\sigma$ (the $i$th cell as $c_i$) in a cyclic manner.
%
%
For the $i$th cell $c_i$ in the total ordering, we assign \EMPH{$\gamma(c_i)$} to be an arbitrary ball $D$ with $D \in \hit(v_{c_i})$ and $D \not \in \hit(v_{c_{i-1}})$. 
Since $|\hit(v_{c_i})| = |\hit(v_{c_i-1})| = k$ and $\hit(v_{c_i}) \neq \hit(v_{c_i-1})$ (by the definition of cell), such a $D$ must always exist, and thus we can assign $\gamma(c_i)$ to all cells.

We now show that this definition of $\gamma$ is sufficient to ensure that each distinct cell $c$ of size $k$ has a unique encoding. 
First, we use Lemma \ref{lma:ancestor_subsets} to show that the representative vertices of cells with the same $\alpha$ and $\beta$ cannot form ancestor relationships in $T_{D_1}(V_k(D_1,D_2))$.

\begin{observation} \label{obs:no_ancestor}
    For any two distinct cells $c,c' \in  C_k(D_1, D_2)$, $v_c$ cannot be an ancestor of $v_{c'}$ in $T_{D_1} (V_k(D_1,D_2)$ or $T_{D_2} (V_k(D_1,D_2)$.
\end{observation}
\begin{proof}
    First consider $T_{D_1} (V_k(D_1,D_2)$. If $v_c$ were an ancestor of $v_{c'}$, then $v_c$ would be in $\SP(D_1,v_{c'})$. This would imply that $\hit(v_{c'}) \subseteq \hit(v_c)$ by Lemma \ref{lma:ancestor_subsets}. However, since $c$ and $c'$ have depth $k$, $|\hit(v_{c})| = |\hit(v_{c'})| = k$. 
    The subset relationship would hence require that $\hit(v_{c}) = \hit(v_{c'})$, and thus $c$ and $c'$ are the same class, which is a contradiction. A symmetric argument applies for $T_{D_2} (V_k(D_1,D_2)$.
\end{proof}

This means that for any $c,c'$ with the same $\alpha$ and $\beta$, $\SP(\alpha(v_c), v_c) \not \subseteq \SP(\alpha(v_c),v_{c'})$. 

\begin{restatable}{lemma}{encoding}\label{lma:encoding}
        For $k\geq 3$, every cell $c$ of depth $k$ has a unique encoding $\langle \alpha(c), \beta(c), \gamma(c) \rangle  \in \cR^3$.
\end{restatable}

\begin{proof}
    If two cells have the same depth and $\alpha$ and $\beta$, they are both indexed by the same traversal order $\sigma$. 
    Consider two cells $c_i,c_j$ such that $\langle \alpha(c_i) , \beta(c_i) , \gamma(c_i) \rangle = \langle \alpha(c_j), \beta(c_j), \gamma(c_j) \rangle = (D_1,D_2,D_3)$ and $|\hit(v_{c_i})| = |\hit(v_{c_j})| = k$. 
    
    We need to consider how $c_i$ and $c_j$ were assigned the same $\gamma$.
    Recall that we force adjacent cells $c_{i-1},c_{i}$ to not just have different $\gamma$ assignments, but that the $\gamma(c_i)$ not contain $\gamma(c_{i-1})$. 
    This forces $v_{c_{i-1}}$ and $v_{c_{j-1}}$ to not be contained by $D_3$. 
    %
    %
    Let \EMPH{$a$} be the lowest common ancestor of $v_{c_{i-1}}$ and $v_{c_{j-1}}$ in $T_{D_1} (V_k(D_1,D_2)$. Similarly, let \EMPH{$b$} be the lowest common ancestor of $v_{c_{i-1}}$ and $v_{c_{j-1}}$ in $T_{D_2} (V_k(D_1,D_2)$. 
    We consider the cycle \EMPH{$F$} formed by the union of the following shortest paths: $\SP(a,v_{c_{i-1}})$ to $\SP(v_{c_{i-1}}, b)$ to $\SP(b,v_{c_{j-1}})$ to $\SP(v_{c_{j-1}}, a)$.
    Since $G$ is planar, $F$ describes the boundary of some region \EMPH{$R(F)$}. 
    Since $i-1 < i < j-1 < j$, by the definition of our traversal one of $v_{c_i}$ or $v_{c_j}$ lies in $R(F)$ and the other lies outside. 
    (We give an example $R(F)$ as a gray region in Figure \ref{fig:encoding_crossing}.) This means that, for $D_3$ to contain both $v_{c_i}$ and $v_{c_{j}}$, $D_3$ must also contain some vertex $u'$ of $F$. 
    Assume WLOG
    that $u'$ lies on $\SP(D_3,v_{c_j})$. 
    Since $u'$ lies on $F$, it must be on one of $\SP(D_1,v_{c_{i-1}}),\SP(D_1,v_{c_{j-1}}), \SP(D_2,v_{c_{i-1}}), $ or $ \SP(D_2,v_{c_{i-1}})$. 

    \begin{itemize}
    \item If $u'$ is in $\SP(D_1,v_{c_{i-1}})$ or $\SP(D_1,v_{c_{j-1}})$, then we know that $D_3$ must be further from $u'$ than $D_1$ since $D_3$ cannot contain $v_{c_{i-1}}$ or $v_{c_{j-1}}$. However, since $\alpha (v_{c_j}) = D_1$, $D_3$ must be closer to $u'$ than $D_1$ since $u'$ is on $\SP(D_3,v_{c_j})$ and $\alpha(v_{c_j})$ is the furthest ball from $v_{c_j}$ which still contains it. 
    (See left subfigure of Figure \ref{fig:encoding_crossing}.)

    \item 
    Similarly, if $u'$ is in $\SP(D_2,v_{c_{i-1}})$ or $\SP(D_2,v_{c_{j-1}})$, then $D_3$ must be further from $u'$ than $D_2$ since $D_3$ cannot contain $v_{c_{i-1}}$ or $v_{c_{j-1}}$. However, since $\beta (v_{c_j}) = D_2$, $D_3$ must be closer to $u'$ than $D_2$ since $u'$ is on $\SP(D_3,v_{c_j})$ and $\beta(v_{c_j})$ is the second furthest ball from $v_{c_j}$, while $D_3$ is at best third furthest. 
    (See right subfigure of Figure \ref{fig:encoding_crossing}.)
    \end{itemize}
    
    \begin{figure}[ht]
    \centering
    \includegraphics[width=0.4\linewidth,page=10]{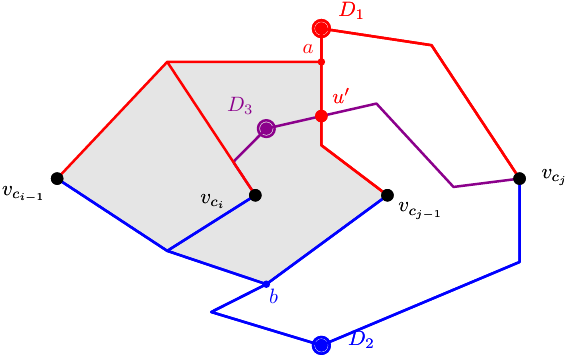}
    \hspace{.1\linewidth}
    \includegraphics[width=0.4\linewidth,page=11]{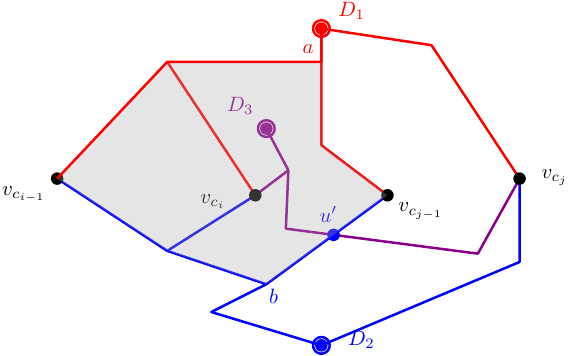}
    \caption{\internallinenumbers
    If two cells $c_i,c_j$ have the same $\gamma(c_i) = \gamma(c_j) = D_3$, then there must $D_3$ must contain a vertex $u'$ of $\gamma_{\alpha,\beta}(v_{c_{j-1}})$ or $\gamma_{\alpha,\beta}(v_{c_{j-1}})$. 
    The region $R(F)$ corresponding to the simple cycle $F$ is given in gray.
    \\
    Left: $u'$ is in $\SP(D_1,v_{c_{j-1}})$. 
    Right: $u'$ is in $\SP(D_2,v_{c_{j-1}})$.}
    \label{fig:encoding_crossing}
    \end{figure}

    We get a contradiction in both cases, which means that $\gamma(c_i)$ cannot equal $\gamma(c_j)$, and thus every cell $c$ of the same depth $k \geq 3$ must have a unique encoding $\langle \alpha(c), \beta(c) , \gamma(c) \rangle$.
\end{proof}

\subsection{Depth-3 Cells}
\label{S:depth-3-cells}

Now that we have a unique encoding $\langle \alpha(v_c), \beta(v_c), \gamma(v_c) \rangle$ for every cell $c$ of depth $k$, we can proceed to give a bound on $|\cell_3|$. 

%
For each $c$ in $\cell_3$, we consider its encoding $\langle \alpha(v_c), \beta(v_c), \gamma(v_c) \rangle$. Since $c$ only has depth 3, we know that these are simply the balls of $\hit(v_c)$, placed in order of distance from $v_c$. Our first step is to select one of the three balls of $\hit(v_c)$ as \emph{$\base(v_c)$}. Next, we iterate through every ball $D$ in $\cR$ and count the number of cells that have been assigned $D$. We accomplish this by constructing a planar graph $G_D$ on the neighbors of $D$ in the Voronoi partition graph $H$. Summing over all $D$, this yields a total count of $O(|\cR|)$ cells. We start with an important lemma regarding the order of distances to balls from vertices that lie on a shortest path.

\begin{lemma}
\label{lma:overtaking}
    Given any ball $D$ of $\cR$ , a vertex $v$, some vertex $v'$ on $\SP(x_D,v)$,
    and a distinct ball $D' \in \cR$, 
    if $d(x_{D'},v') \leq d(x_D, v')$, then for all $v'' \in \SP(v',v)$, $d(x_{D'},v'') \leq d(x_D, v'')$.
\end{lemma}
\begin{proof}
    For any $v'' \in \SP(v',v)$, because $v'$ is between $x_D$ and $v''$ on $\SP(x_D,v)$,  $d(x_D,v'') = d(x_D,v') + d(v',v'')$. 
    We also know that $d(x_{D'}, v'') \leq d(x_{D'},v') + d(v',v'')$. Since $d(x_{D'},v') \leq d(x_D,v')$, we can simply plug in to get $d(x_{D'}, v'') \leq d(x_{D},v') + d(v',v'') = d(x_D,v'') $.
\end{proof}


\begin{observation}\label{obs:stays_in_3_cells}
    For a cell $c$ and ball $D' \in \hit(v_c)$, all vertices $v$ in $\SP(x_D,v_c)$ must lie in the Voronoi cell of a ball in $\hit(v_c)$.
\end{observation}
\begin{proof}
    Consider a ball $D' \not \in \hit(v_c)$ for which a vertex $v$ of $\SP(x_D,v_c)$ lies in the Voronoi cell $F_D$. If this is the case, then all the vertices after $v$ in $\SP(x_D,v_c)$ will be closer to $D'$ than they are to $D$, and since $D$ contains $v_c$, $D'$ must also then contain $v_c$, which is a contradiction.
\end{proof}

This means that the shortest paths from the centers of the balls hit by $v_c$ to the vertex $v_c$ must stay within the Voronoi cells of $\hit(v_c)$. This observation, paired with Lemma~\ref{lma:overtaking}, allows us to define \EMPH{$\base(v_c)$} based on $\SP(x_{\alpha(v_c)}, v_c)$. We consider two cases:

\begin{itemize}
\item
\textit{Case 1: Some vertex $v$ in $\SP(x_{\alpha(v_c)},v_c)$  lies in $F_{\beta(v_c)}$.~} 
Assign $\base(v_c) \gets \beta(v_c)$. 

\item
\textit{Case 2: No vertex $v$ in $\SP(x_{\alpha(v_c)},v_c)$ lies in $F_{\beta(v_c)}$.~} 
Assign $\base(v_c) \gets \gamma(v_c)$. 
\end{itemize}

\begin{figure}[ht]
    \centering
    \includegraphics[width=0.3\linewidth,page=19]{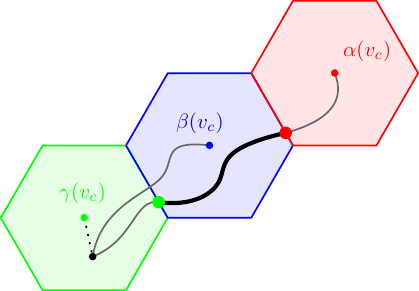}
    \hspace{.1\linewidth}
    \includegraphics[width=0.3\linewidth,page=18]{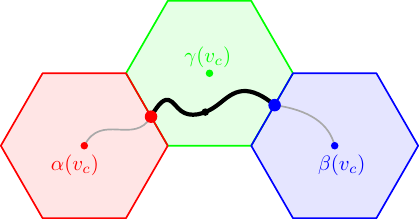}
    
    \caption{\internallinenumbers
    Showing the subwalk of Lemma~\ref{lma:subwalk_new} and how to select $\base$.
    Left: Case 1.
    Right: Case 2.}
    \label{fig:subwalk_defintion}
\end{figure}

Now that we have assigned every cell to a ball, we consider a ball $D$ and the set \EMPH{$\resident(D)$}, which is the set of cells $c$ with $\base(c) = D$. We will consider the cells of $\resident(D)$ in relation to the neighbors of $D$ within the Voronoi partition graph $H$.

\begin{lemma} \label{lma:subwalk_new}
    For a ball $D$ and every cell $c$ of $\resident(D)$, there is a subwalk $\pi_c$ of $\pab[v_c]$ with the (nonempty set of) interior vertices lying in $D$ and
    the end vertices lying in distinct neighbors $D',D''$ of $D$ in $H$. 
\end{lemma}
\begin{proof}
    If cell $c$ is a \textit{case-1} cell, we have that a vertex $v$ in $\SP(x_{\alpha(v_c)},v_c)$ which lies in  $F_{\beta(v_c)}$. By Lemma~\ref{lma:overtaking}, no vertex of $\SP(v,v_c)$ can lie in $F_{\alpha(v_c)}$. Additionally, $v_c$ must lie in $F_{\gamma(v_c)}$ by definition of $\alpha$ and $\beta$. Let $u$ be the first vertex of $\SP(v,v_c)$ which lies in $F_{\gamma(v_c)}$. Let $v'$ be last vertex of $\pab[v_c]$ which lies in $F_{\alpha(v_c)}$. The subwalk $\pab[v_c] [v',u]$ satisfies this condition.

    If cell $c$ is a \textit{case-2} cell, we have no such vertex and $\SP(x_{\alpha(v_c)},v_c)$ lies entirely in $\alpha(v_c)$ or $\gamma(v_c)$.  Let $u$ be the last vertex of $\SP(x_{\alpha(v_c)},v_c)$ which lies in $F_{\alpha(v_c)}$ and let $u'$ be the last vertex  of $\SP(x_{\beta(v_c)}, v_c)$ that lies in $F_{\beta(v_c)}$. We then can take our subpath to be $\pab[v_c][u,u']$, and this will satisfy the condition.
\end{proof}

We now consider the neighbors of $D$ in the Voronoi support $H$, \EMPH{$N_H(D)$}. In Lemma~\ref{lma:subwalk_new}, the subwalk $\pi_c$ of cell $c$ in $\resident(D)$ runs between the two balls of $\hit(v_c) \setminus D$, both of which must be in $N_H(D)$ since $H$ is simply the contraction of the Voronoi partition of $G$. 
We can now define a graph \EMPH{$G_D$} on $N_H(D)$ to count the number of cells in $\resident(D)$. For every cell $c$ of $\resident(D)$, add an edge \EMPH{$e_c$} between the balls of $\hit(v_c) \setminus D$. Since the underlying graph $G$ is planar, we can sketch a drawing of $G_D$ as follows:

To draw the curve \EMPH{$\zeta_c$} corresponding to cell $c$ of $\resident(D)$, we first consider the subpath $\pi_c$ that we get from Lemma~\ref{lma:subwalk_new}, using a plane drawing of $G_D$. 
Let $v$ be the first endpoint of $\pi_c$ lying in a ball $D_1$ and $v'$ be the second endpoint of $\pi_c$ lying in a ball $D_2$. 
We first draw the center $x_{D_1}$ in its location on any plane drawing of the underlying graph $G$. We then sketch $\zeta_c$ on $\SP(x_{D_1, v})$. 
This is followed by sketching $\pi_c$, and then finally sketching the walk $\SP(v', x_{D_2})$. Since this sketch is over a path, we will refer to this walk in $G$ as \EMPH{$\pi(\zeta_c)$}.

\begin{figure}[ht]
    \centering
    \includegraphics[width=0.25\linewidth,page=20]{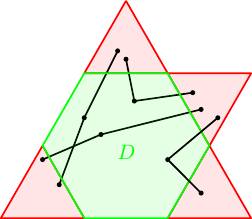}
    \hspace{.1\linewidth}
    \includegraphics[width=0.25\linewidth,page=21]{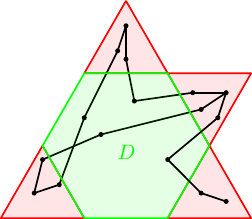}
    \caption{\internallinenumbers
    Left: a set of subwalks of cells in $\resident(D)$. Right: the walks $\pi(\zeta_c)$ for each $c \in \resident(D)$.}
    \label{fig:dwrawing.}
\end{figure}

\begin{lemma}\label{lma:no_crossing}
    For a given ball $D$ in $\cR$ and two distinct cells $c$ and $c'$ in $\resident(D)$, if $\pi(\zeta_c)$ and $\pi(\zeta_{c'})$ share a common vertex $v$, then $e_c$ and $e_{c'}$ share a common endpoint $D'$.
\end{lemma}
\begin{proof}
    Let $e_c = (D_1,D_2)$ and $e_{c'} = (D_3,D_4)$. For $c$ $(c')$, the common vertex $v$ must lie on a shortest path from either $x_{D_1}$ or $x_{D_2}$ ($x_{D_3}$ or $x_{D_4}$) to $v_c$ ($v_{c'}$) by construction. Without loss of generality, assume $x_{D_1}$ and $x_{D_3}$. One of $x_{D_1}$ or $x_{D_3}$ must be at least as close as the other to $v_c$. Without loss of generality, assume $x_{D_1}$. By Lemma~\ref{lma:overtaking}, this implies that $D_1$ contains $v_{c'}$, and is thus in $\hit(v_{c'})$. We know that $D$ is contained within $\hit(v_{c'})$, as we only ever assign $\base(v_{c'})$ to a member of $\hit(v_{c'})$. Thus, in order for $c'$ to be depth 3, $D_1$ must be either the same as $D_3$ or $D_4$. Thus, $e_c$ and $e_{c'}$ have a common endpoint.
\end{proof}

    Lemma~\ref{lma:no_crossing} implies that $\zeta_c$ and $\zeta_{c'}$ only cross if $e_c$ and $e_{c'}$ share an endpoint. Since the only crossings occur between edges incident on the same vertex, $G_D$ can be redrawn with these crossings removed. This immediately implies the following:

\begin{lemma}\label{g_d_planar}
    $G_{D}$ is a simple planar graph.
\end{lemma}





\begin{restatable}{lemma}{depthlinear}\label{lma:depth_3_linear}
        There are $O(|\cR|)$ cells of depth at most 3.
\end{restatable}

\begin{proof}
    We first quickly argue that there is $O(n)$ cells of depth at most 2. By the definition of cell, there is at most $n$ cells of depth 1, one for each ball in $\cR$. By the existence of a planar dual support (\Cref{thm:dual_support}) on $\cR$ with respect to all vertices of $G$, we know that for every cell of depth 2 corresponding to vertices $v$ incident on balls $R$ and $R'$, the connectivity of $\cR_v$ requires that $R$ and $R'$ share an edge. Since the dual support is planar, there is at most $3 |\cR|-6$ such pairs, and thus at most $O(n)$ cells of depth 2.
    
    For depth-3 cells, 
    every edge of $G_D$ corresponds to a resident depth-3 cell $c$ of $D$, so $|\resident(D)| \leq |E(G_D)|$. By Lemma~\ref{g_d_planar}, $|\resident(D)| \leq 3 |V(G_D)|$. 
    By charging the depth-3 cells to the total number of edges across all $G_{D}$ for all $D \in \cR$ using the planarity of $G_D$,
    $|\resident(D)| \leq 3|N_H(D)|$. 
    Summing over all neighborhoods in the contracted Voronoi diagram of $\cR$ in $G$, we know the sum of degrees is at most $6 |\cR|$. 
    Thus,
    $|\cell_3| \leq 18 | \cR |$. 
\end{proof}


\bibliography{support}

\appendix

\section{Missing proofs}
\label{S:scaffold}


\subsection{Weighted $r$-Dominating Set: Proof of Theorem~\ref{thm:weighted-dom}}

\begin{theorem}\label{thm:constant_factor_set_cover_planar}
    Given a set of balls $\cR$ in a planar graph $G$, a subset of vertices $U \subseteq V(G)$, and a weight function $w : \cR \to \mathbb{R}_{\geq 0}$, an $O(1)$ approximation for the minimum weight coverage of $U$ by balls of $\cR$ can be computed in polynomial time.
\end{theorem}
\begin{proof}
    Theorem \ref{thm:shallow_cell_planar} implies the balls in $\cR$ in $G$ have $O(n k^3)$ shallow-cell complexity.
    Chan \etal~\cite{cg-spfge-2015} showed that for any set system with shallow cell complexity $O(n \poly k)$, an $O(1)$-approximation for weighted set cover can be computed in polynomial time.
\end{proof}

Theorem~\ref{thm:constant_factor_set_cover_planar} can now directly be applied to the $r$-dominating set problem.

\begin{proof}[Proof of Theorem~\ref{thm:weighted-dom}]
    Consider the set of distance balls $\cR = \{ B(v,r) : v \in V_G \}$. If a vertex $v'$ is within distance $r$ of $v$, then $v'$ is within the ball of $B(v,r) \in \cR$. 
    This implies that the set cover instance $(V_G, \cR)$ is exactly equivalent to $r$-dominating set in $G$.
    Therefore, if we apply the algorithm from Theorem~\ref{thm:constant_factor_set_cover_planar} to $(V_G, \cR)$, then we will get an $O(1)$-approximation for $r$-dominating set in polynomial time.
\end{proof}

\section{Existence of Intersection Support}\label{sec:intersection}

We now shift focus to a more general notion of support. 
For other more complicated applications, sometimes we need to work with set systems that have two different types of sets.
We define an \EMPH{intersection system} $(U, \cR, \cB)$ as a set system where the elements are the sets of $\cR$ and the family of sets is $\cB$, and we consider a set $R\in \cR$ to be a ``member'' of set $B \in \cB$ if there is at least one element of $U$ is in the intersection $R \cap B$. 
We define an \EMPH{intersection support} to be a support graph of an intersection system $(U,\cR,\cB)$. In other words, an \emph{intersection support} is an undirected graph \EMPH{$\HH(\cR, \cB)$} where $\cR$ and $\cB$ are two families of subsets of some universe $U$ and $\HH(\cR,\cB)$ is a graph on $\cR$ such that for every set $B \in \cB$, all sets in $\cR$ that contain a common element with $B$ form a connected subgraph of $\HH(\cR,\cB)$. 

\begin{figure}[h!]
    \centering
    \includegraphics[page=2, width=0.7\linewidth]{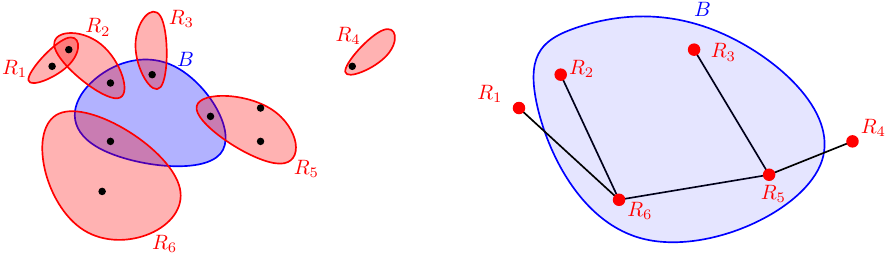}
    \caption{\internallinenumbers
    An intersection system with an intersection support. The red regions intersecting blue set $B$ form a connected subgraph in the~support.
    }
    \label{fig:inter_support_def}
\end{figure}

Our notion of \emph{intersection support} is equivalent to Raman and Ray's notion of planar support for intersection hypergraphs~\cite{RamanRay2020}, except we do not require that our support graphs be planar.
In our case, we are looking at the situation where $\cR$ and $\cB$ are both sets of balls that lie in some underlying graph $G$. It is important to note that this is a generalization of both dual and primal support for the case of balls, since we can always consider the case where $\cR$ is a set of radius-0 balls on every vertex to get a primal support and the case where $\cB$ is the set of radius-0 balls on every vertex to get a dual support. 

\paragraph{Undirected case.}
We reuse the same construction of $G'$, but in this case adding nodes for both $\cR$ and $\cB$.
It is easy to show that if $G$ (and our augmentation $G'$) is undirected, then the above trick of contracting the Voronoi partition with respect to $\cR$ works: Consider each $\cR$ node to be the center of a ball with radius $2r_{\max}$. If a ball $R$ from $\cR$ and $B$ from $\cB$ share a vertex in $G$, then the ball centered at $x_R$ in $G'$ of radius $2r_{\max}$ must contain $x_B$. 
By Lemma \ref{lma:dual_is_a_support}, we know that all $R' \in \cR$ containing $x_B$ must induce a connected subgraph in $\HH$, satisfying the definition of intersection support.

\medskip
In this section we show how to extend the technique to obtain an intersection support which also preserves forbidden minors even when $G$ is directed.

\begin{restatable}[Intersection Support]{theorem}{intersectionsupport}\label{thm:intersection_support}
    Any system of balls $(\cR,\cB)$ on an underlying graph $G$ has an intersection support $\HH(\cR,\cB)$ such that if $G$ does not contain a (contraction-)minor $K$ with minimum vertex degree two, then neither does $\HH(\cR,\cB)$.
\end{restatable}

\subsection{Constructing an intersection support}

To construct an intersection support for a directed graph $G$, we first construct $G'$, using $\cR \cup \cB$ as our balls for the node set. This means that we create a node for each $D \in \cR \cup \cB$ and delete all vertices that are not contained in any such $D$. We will additionally label each vertex (and node) with whether it is reachable from a ball in $\cR$ (\textbf{\EMPH{red}}) or if it is \textit{only} reachable from balls in $\cB$ (\textbf{\textcolor{Blue}{\emph{blue}}}). Considering only the red vertices at first, we apply the same procedure of constructing the Voronoi partition with respect to $\cR$ and contracting each cell. We refer to this graph as $\HH$. $\HH$ should now only contain blue vertices, blue nodes and red nodes. 
To define edge weights of $\HH$, we will say that edges between two red nodes have infinite weight, whereas all other edges inherit the minimum weight whenever two edges merge due to an edge contraction. 
From here, we flip the orientations of every edge in $H$, including those incident on nodes. We then construct a Voronoi partition with respect to the red nodes $\cR$, and contract every cell, in this case yielding a graph \EMPH{$H'$} only on the red nodes.

We first show that $H'$ is a support graph, as was done with Lemma \ref{lma:dual_is_a_support} for the dual system.

\begin{lemma}\label{lma:intersection_is_a_support}
     For every blue ball $B \in \cB$ in a system of balls $(\cR,\cB)$ on an underlying graph $G$, the set of red balls $\cR_B \subseteq \cR$ that intersects $B$ nontrivially
     must induce a connected subgraph in $H'$.
\end{lemma}

\begin{proof}
    Assume that $|\cR_B| \geq 2$, as otherwise this holds trivially.
    Consider the closest red vertex $u$ to $B$ (smallest distance $d(x_B,u)$). By the construction of $\HH$, $u$ was contracted into some red node $R \in \cR$. In fact, $R \in \cR_B$, $|\cR_B| = 0$ by the definition of a red vertex. 
    Since we always choose the minimum weight edge when merging, we know that $x_B$ must be in the Voronoi cell of $R$ in $\HH$. It suffices to show that the remaining $R' \in \cR_B \setminus\{R\}$ must have a path to $R$ within $H'[\cR_B]$.

    Consider a vertex $v$ in $R' \cap B$, and the shortest paths $\SP(x_{R'},v)$ and $\SP(x_B,v)$. We know that every vertex $v'$ in $\SP(x_{R'},v)$ must be in the $G'$ Voronoi cell of $R'$. We claim that every vertex $v'$ in $\SP(x_B,v)$ is either a red vertex in the $G'$ Voronoi cell of a ball in $\cR_B$ or a blue vertex in the $\HH$ Voronoi cell of a ball in $\cR_B$. This lets use the union of $\SP(x_{R'},v)$ and $\SP(x_B,v)$ to trace a path in $H'$ from $R'$ to $R$.

 If $v'$ is in $\SP(x_B',v)$, then, by definition of shortest path, $d(x_B,v') \leq d(x_B,v)$ and thus $v'$ is in $B$. With this in mind, we can split into cases based on whether $v'$ is a red vertex or a blue vertex.

    \begin{itemize}

\begin{figure}[ht]
    \centering 
    \includegraphics[width=0.25\linewidth,page=5]{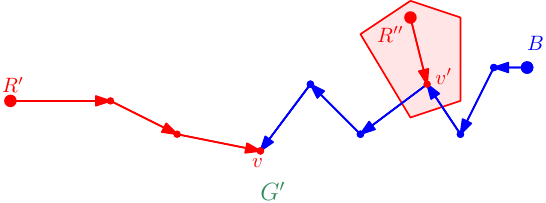}
\hspace{.07\linewidth}
\includegraphics[width=0.25\linewidth,page=6]{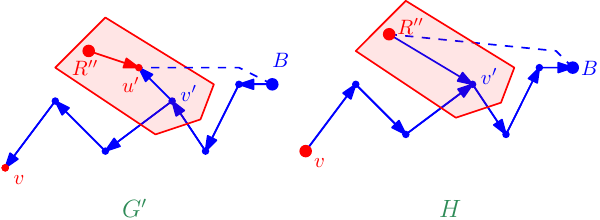}
    \caption{
    Left: a red vertex $v'$ of $\SP(c_B,v)$ must lie in the Voronoi cell of a red ball $R'$ intersecting with $B$. Right: A blue vertex $v'$ of $\SP(c_B,v)$ must lie in the $\HH$ Voronoi cell of $R'$ intersecting with $B$.}
    \label{fig:intersection_lemma}
\end{figure}

\item If $v'$ is red then $v'$ must be in the Voronoi cell (in $G'$) of some red ball $R''$ that contains $v'$. Since $v'$ is also in $B$, this means that $R'' \cap B \neq 0$, which implies that $R'' \in \cR_B$.

\item   
    If $v'$ is blue, there must be some red vertex $u'$ which is of the smallest distance, $d(v',u')$, among all red vertices. Since $v$ is a red vertex, we know that $d(v',u') \leq d(v',v)$, hence:
    \[
    d(x_B, u') \leq d(x_B,v')+d(v',u') \leq  d(x_B,v') + d(v',v)  = d(x_B,v) \leq r_{\max}
    \]
    This implies that $u'$ must be in $B$.
    
    Since $u'$ is red, $u'$ must be in some Voronoi cell of a red ball $R''$ (in $G'$). After the orientation flip and the construction of $\HH$, $u'$ is contracted into $x_{R''}$, and thus $x_{R''}$ must be the closest red node to $v'$. This means $v'$ is in the $R''$ Voronoi cell in $\HH$. Since $u' \in R'' \cap B$, $R'' \in \cR_B$, thus every vertex in $\SP(x_B,v)$ is contracted into a red node which is in $\cR_B$.
    \end{itemize}

    This means that we can simply perform a walk from $R$ to $R'$ in $H'$ by going to the red ball corresponding to each vertex along $\SP(x_{R'},v)$ and then backwards along $\SP(x_B,v)$. Thus, $\cR_B$ must form a connected subgraph in $H'$.
\end{proof}


\begin{proof}[Proof of Theorem \ref{thm:intersection_support}]

    Construct $H'$ described above as our intersection support. 
    By Lemma \ref{lma:intersection_is_a_support}, we know that $H'$ satisfies the definition of our intersection support, so it remains to show that $H'$ contains no (contraction)-minor with minimum degree two that is not also a minor of $G$.
    The construction of $H'$ begins with the construction of a graph $G'$ almost equivalent to the one from Theorem \ref{thm:dual_support}, but for set system $(G, \cR \cup \cB)$. 
    From the proof of Theorem \ref{thm:dual_support}, we know that this $G'$ contains no minor $K$ of minimum degree two. 
    Since we only perform contractions and vertex deletions after constructing $G'$, we know that $H'$ is a minor of $G'$, and thus also contains no minor $K$.
\end{proof}

To conclude, any system of balls admit an intersection support which contains the same forbidden minors as the underlying graph. Both the dual and intersection supports can be computed in polynomial time. 
The primary bottleneck is the construction of the Voronoi partitions, which Erwig~\cite{erwig2000graph} showed to be computable in $O(m + n \log n)$ time for graphs with $m$ edges and $n$ vertices.

\end{document}